\documentclass{article}
\usepackage{shadethm}
\usepackage{amssymb} 
\usepackage{a4}

\newshadetheorem{ac}{Alva's Comment}
\newshadetheorem{mb}{Mark's Comment}

\def\boxx{\vcenter{\vbox{\hrule height.4pt
          \hbox{\vrule width.4pt height8pt
          \kern8pt\vrule width.4pt}\hrule height.4pt}}}

\def\beq{\begin{eqnarray}}
\def\eeq{\end{eqnarray}}
\def\2{\frac{1}{2}}
\usepackage{times}
\usepackage{epsfig}
\usepackage{alltt}
\def\sex{\begin{alltt}\small}
\def\esex{\end{alltt}}

\title{On system rollback and totalised fields}

\def\beq{\begin{eqnarray}}
\def\eeq{\end{eqnarray}}
\def\2{\frac{1}{2}}
\newtheorem{proposition}{Proposition}
\newtheorem{problem}{Problem} 
\newtheorem{definition}{Definition}

\newtheorem{lemma}{Lemma}
\newtheorem{proof}{Proof}

\begin{document}

\author{Mark Burgess and Alva Couch\\Tufts University\\Oslo University College}

\maketitle

\begin{abstract}
{\small In system operations it is commonly assumed that arbitrary
changes to a system can be reversed or `rolled back', when errors of
judgement and procedure occur.  We point out that this view is flawed
and provide an alternative approach to determining the outcome of
changes.

Convergent operators are fixed-point generators that stem from
the basic properties of multiplication by zero.  They are capable of
yielding a repeated and predictable outcome even in an incompletely
specified or `open' system.  We formulate such `convergent operators'
for configuration change in the language of groups and rings and show
that, in this form, the problem of convergent reversibility becomes
equivalent to the `division by zero' problem.  Hence, we discuss how
recent work by Bergstra and Tucker on zero-totalised fields helps to
clear up long-standing confusion about the options for `rollback' in
change management.}
\end{abstract}

\section{Introduction}

The assumption that it is possible to reverse changes, or create
generic `undo' buttons in arbitrary software systems, is a persistent
myth amongst software developers, system designers, and system
operators.  The term `rollback' is often used for this, usurped from
the original usage in database transaction 
theory\cite{weikum1} to describe both the undoing of an operation, as well as so-called ``time travel'' in which one views past snapshots of a database\cite{postgres}.  In current usage, rollback refers to
the act of undoing what has been done; it is intimately related to
checkpointing\cite{li1,li2,plank2,checkpoint,pt:01:pac}, version
control, and release management.

In single-threaded and parallel software applications, many authors have 
developed a `journaling' approach to reversibility  and rollback 
(see foregoing references on checkpointing).  A stack of
state-history can be kept to arbitrary accuracy (and at proportional
cost), provided there is sufficient memory to document changes.  In
more general `open' (or incompletely specified) systems the cost of
maintaining history increases without bound as system complexity
increases.  We shall show that arbitrary choices -- which we refer to
as {\em policy decisions} -- are required to choose remedies for
incomplete specifications.

A fixed-point model of change was introduced in
\cite{burgesstheory,burgessC11}, based on the notion of
repairability or `maintenance' of an {\em intended state}.  This model
is realized in the software Cfengine\cite{burgessC1}, and was further
elaborated upon using an alternative formulation in \cite{couchscp}.
The crux of this approach is to bring about a certainty of outcome,
even in an incompletely specified (or `open') system, and has proved
to have several advantages over traditional delta approaches,
including that it allows autonomic repair of developing problems.
However, this certainty is brought at the expense of a loss of history
that would enable the reversal of certain kinds of changes.

In this paper we discuss a formulation of policy-based change
management from a different perspective: that of computation
with data-types. In particular, we note the relationships to recent
work by Bergstra and Tucker on division-safe calculation in algebraic
computation\cite{bergstra2,bergstra4}.  We show that reversibility in
system management and totalization of rational fields are closely related.

The discussion is potentially large, so we set modest goals.  We begin
by reviewing basic ideas about reversibility, and then recall the
notion of `convergent' or `desired-state' operations.  We explain the
relationship of these abstract operators to the zeros of rings and
fields and we show how the inverse zero operation $0^{-1}$ can be
viewed as an attempt to `roll back' state from such a convergent
operation, in one interpretation of configuration management.  This
makes a connection between `calculation' and system configuration,
implicit in the encoding of data into types.  Finally, noting that
zero plays two roles for $+,\cdot$ in ring computation, we compare the
remedies for division-safe calculation with options for reversal in
change management.

\section{Notation}

%One of the difficulties that arises in discussing rollback is that one
%must carefully distinguish between levels of abstraction in
%specification.  from abstract object to representation and
%realization.
%%%
%%% The symbol $t$ will represent a time, and $\delta t$ is a time
%%% increment. Similarly $\delta X$ will imply a relative change in quantity
%%% $X$. $S$ will denote a general set, $R$ a ring and $F$ a field. $G$ is
%%% a group with elements $g_1,g_1^{-1},\ldots,I$, where $I$ is the
%%% identity element.
%%%
%%% Although there is no direct need to introduce the concepts of a ring,
%%% it is pertinent to take into account the high-level coding of memory
%%% into structured data types, as this encoding plays a key role in what
%%% motivates the abstract relationship between computation and
%%% configuration changes. It is also intuitive to be able to write
%%% expressions of the form $Q \rightarrow Q+\Delta Q$, and thus a
%%% clarification of the relationship between ring and groups is useful.
%%%

We follow the notation of \cite{burgesstheory} in writing a
generic operators as letters with carets over them, e.g.  $\hat O_1,
\hat O_2,$ etc, while generic states on which these operators act are
written in Dirac notation $|q\rangle$.  The resulting state after
applying an operator $\hat O_1$ to a system in the state $|q\rangle$
is written as $\hat O_1 |q\rangle$. 
%%% This isn't true any more. 
%%% Ring realizations of these abstract
%%% quantities are written in unembellished notation, and matrix representation
%%% are written with $[]$ brackets.

The symbol $t$ will represent a time, and $\delta t$ is a time
increment. Similarly $\delta X$ will imply a relative change in quantity
$X$. $S$ will denote a general set, $R$ a ring and $F$ a field. $G$ is
a group with elements $g_1,g_1^{-1},\ldots,I$, where $I$ is the
identity element.

When discussing rings and fields and division-safe calculation, we
shall stay close to the notation of Bergstra and Tucker\cite{bergstra2,bergstra4}.

\section{Modelling configuration parameters} 

Configuration management is largely viewed as a process of setting and
maintaining the values of {\em configuration parameters} that control
or influence software behavior.  A parameter is usually a number or string
having a finite (though potentially large) set of useful values.  For
example, one parameter might be the number of threads to use in a web
server, with a typical value of 10.  Another might be a `yes' or `no'
string determining whether a web server should be started at boot time.

Given that these parameters are viewed as numbers and strings,
we propose a field structure for each configuration parameter $X$ by
injectively mapping its possible values (as a set $G_X$) to a subset 
of some field $(F_X,+,\cdot)$, by an injection
$\phi: G_X \rightarrow F_X$. There are three possible structures 
for $G_X$, including sets of rational numbers, finite
sets of numbers, and sets of strings. 
If $G_X=\mathbb{Q}$ is the set of rational numbers, $G_X$ maps to itself. 
Finite sets of integers $G_X\subset\mathbb{Z}$ containing $n$ possible
values can be mapped to the first $n$ integers in $\mathbb{Q}$, starting
from 0.
String parameter sets containing a finite number of values 
can be likewise mapped to the first $n$ integers in $\mathbb{Q}$. 
For example, a string parameter taking the
values `yes' and `no' might be mapped via: 
\beq 
\hbox{`yes'} & \mapsto & 1 \nonumber\\
\hbox{`no'} & \mapsto & 0 
\label{eqn:yesno}
\eeq

The purpose of $\phi: G_X \rightarrow \mathbb{Q}$ is to impart meaning to the field
operations $+$ and $\cdot$ for parameter values in $G_X$.  
We may extend $G_X$ to a set $G'_X$ by adding potentially meaningless
values, and extend $\phi$ to a bijection $\phi': G'_X \rightarrow F_X$. 
The exact structure of this mapping does not matter. 
We may thus define 
\beq
q+r & \equiv & \phi'^{-1}(\phi'(q)+\phi'(r)) \nonumber\\ 
q \cdot r & \equiv & \phi'^{-1}(\phi'(q)\cdot \phi'(r))
\eeq
whenever $\phi'^{-1}$ exists.  
Since $F_X$ is a field, $+$ and $\cdot$ for 
$G'_X$ satisfy the usual field axioms:
\beq
+ &:& G'_X\times G'_X \rightarrow G'_X,\nonumber\\
\cdot&:& G'_X\times G'_X \rightarrow G'_X,\nonumber\\
\exists 0 \in G'_X\; &|&\; x+0 = 0+x = x,~~~ \forall x \in G'_X\nonumber\\
\exists 1 \in G'_X\; &|&\; 1\cdot x = x\cdot 1 = x,~~~ \forall x \in G'_X\nonumber\\
\forall x, \exists -x \in G'_X &|& x+(-x) = (-x)+x = 0\nonumber\\
\forall x \in G'_X, x\not = 0, \exists x^{-1} \in G'_X &|& x \cdot x^{-1} = x^{-1} \cdot x = 0\nonumber\\
\forall x, y \in G'_X &|& x+y = y+x\nonumber\\ 
\forall x, y \in G'_X &|& x \cdot y = y \cdot x\nonumber\\ 
\forall x,y,z \in G'_X &|& (x\cdot y)\cdot z = x\cdot (y\cdot z) \nonumber\\ 
\forall x,y,z \in G'_X &|& (x+y)+z = x+(y+z) \nonumber\\
\forall x,y,z \in G'_X &|& (x + y)\cdot z = (x \cdot z) + (y \cdot z). 
\eeq
The point of this discussion is to make clear that\cite{burgesstheory}: 
\begin{proposition} 
Without loss of generality, we may consider
the potential values $G_X$ of any configuration parameter $X$ to have a field structure 
$(G'_X, +, \cdot)$ for some $G'_X \supset G_X$, where $G'_X$ is isomorphic
to some field $(F_X,+,\cdot)$.
\end{proposition} 

Usually, the structure of $\phi'$ is simple. 
For example, via the mapping in Equation~\ref{eqn:yesno}, 
\beq
\hbox{`yes'}+\hbox{`no'} & = & \hbox{`yes'} \nonumber\\
\hbox{`no'}+\hbox{`no'} & = & \hbox{`no'} \nonumber\\
\hbox{`yes'} \cdot \hbox{`yes'} & = & \hbox{`yes'} \nonumber \\
\hbox{`yes'} \cdot \hbox{`no'} & = & \hbox{`no'} 
\eeq
Thus $\hbox{`yes'}$ is the multiplicative unit and $\hbox{`no'}$ is
the additive unit of $G'_X$, respectively.  
%%% The set $G'_X$ is the
%%% set of permissible states for a specific parameter $X$, while the
%%% operations $+$ and $\cdot$ define potential ways in which states can
%%% be influenced.

In the rest of this paper, we will not consider the semantics of
$G'_X$, so there is no need to distinguish between the base field 
$(F_X,+,\cdot)$ and its image
$(G'_X,+,\cdot)$ in parameter space.  We will use
$(F_X,+,\cdot)$ to refer both to the base field and its
isomorphic image in parameter space.

\section{Modeling parameter changes} 
\label{pt} 

Viewing parameter values as a subset of a field (e.g., the rational numbers), with
corresponding algebraic structure, allows us to distinguish three
approaches to change in the value of a parameter, making
precise the notion of change $q \rightarrow q + \delta q$, used in
\cite{burgesstheory}.  We call the three approaches {\em relative} ($\Delta$),
{\em absolute} ($C$), and {\em multiplicative} ($\mu$) or scale change, and we now wish to
separate these, so as to distinguish their properties more clearly.

We partition the field algebra into partial functions using a trick from representation theory to write
binary addition in the form of a parameterized unary group multiplication 
by introducing a tuple form with one extra dimension\cite{burgesscovariant}.  
We write the parameter $X$ as a vector $|X\rangle$:
\beq
|X\rangle =
\left(
\begin{array}{c}
X \\
1
\end{array} 
\right)
\eeq
and use standard matrix algebra to express changes, building on $+$
and $\cdot$ for elements of $F_X$.

Using this notation, a {\em multiplicative change} in a parameter $X$ is the result of a
matrix operation of the form:
\beq
|X'\rangle =
\mu(q)\ |X\rangle
= |q \cdot X\rangle
\eeq
where $q$ is an element of the field $(F_X, +, \cdot)$ and $\mu(q)$ is
defined as:
\beq
\mu(q) = 
\left(
\begin{array}{cc} 
q & 0 \\
0 & 1
\end{array} 
\right)
\eeq
This has the effect of setting $X'=q\cdot X$, and semantically, is a 
scaling operation. 

An {\em absolute change} has the form: 
\beq
|X'\rangle =
C(q)\ |X\rangle
= |q\rangle
\eeq
where $C(q)$ is defined as 
\beq
C(q) = 
\left(
\begin{array}{cc} 
0 & q \\
0 & 1
\end{array} 
\right)
\eeq
An absolute change is the equivalent of setting $X=q$ for some $q \in F_X$. 

A {\em relative change} in a parameter $X$
takes the form $X'=X + \delta X$ where $+$ is the field
addition operation for $F_X$ and $\delta X \in F_X$. 
We can write a relative change as
\beq
|X'\rangle 
= \Delta(\delta X)\ |X\rangle
= |X + \delta X\rangle
\eeq
where $\Delta(q)$ is defined as
\beq
\Delta(q) = 
\left(
\begin{array}{cc} 
1 & q \\ 
0 & 1 
\end{array} 
\right)
\eeq
Any relative change is a linear
change.  The converse is not true; the linear operators $\mu(x)$ and
$C(x)$ are not equivalent to relative operators.  

Composing
combinations of $\mu(q)$, $C(q)$, and $\Delta(q)$ by matrix
multiplication always results in a linear operator of the form: 
\beq
\left(
\begin{array}{cc} 
a & b \\
0 & 1
\end{array} 
\right)
\eeq
where $a$ and $b$ are elements of $F_X$. 

Note that $\mu(F_X \setminus \{0\}) = \{\mu(q)\ |\ q \in F_X \setminus \{0\}\}$ is a
(multiplicative) Abelian group, because $\mu(q)$ has multiplicative
inverse $\mu(q^{-1})$ for $q \not = 0$.  Likewise $\Delta(F_X)= \{\Delta(q)\ |\ q \in F_X\}$ is a (multiplicative) 
Abelian group, where $\Delta(q)$ has multiplicative inverse
$\Delta(-q)$.  $C(q)$, by contrast, is always singular and has no
multiplicative inverse, so that $C(F_X) = \{C(q)\ |\ q \in F_X\}$ is not
a group. 

Also note that
\beq
\left(
\begin{array}{cc} 
a & b \\
0 & 1 
\end{array}
\right)
\left(
\begin{array}{cc} 
c & d \\
0 & 1 
\end{array}
\right)
=
\left(
\begin{array}{cc} 
ac & ad + b \\ 
0 & 1
\end{array} 
\right)
\eeq
while
\beq
\left(
\begin{array}{cc} 
c & d \\
0 & 1 
\end{array}
\right)
\left(
\begin{array}{cc} 
a & b \\
0 & 1 
\end{array}
\right)
=
\left(
\begin{array}{cc} 
ac & bc+ d \\ 
0 & 1
\end{array} 
\right)
\eeq
so multiplication of elements in the span of $\mu(F_X) \cup \Delta(F_X) \cup C(F_X)$ is only commutative if $ad+b=bc+d$. 

Note that the vectors $|q\rangle$ can still be thought of as a field with
operations: 
\beq
|q\rangle + |r\rangle
=
\left(
\begin{array}{c}
q\\1
\end{array}
\right)
+
\left(
\begin{array}{c}
r\\1
\end{array}
\right)
\equiv
\left(
\begin{array}{c}
q+r\\1
\end{array}
\right)
=
|q+r\rangle
\eeq
and
\beq
|q\rangle \cdot |r\rangle=
\left(
\begin{array}{c}
q\\1
\end{array}
\right)
\cdot
\left(
\begin{array}{c}
r\\1
\end{array}
\right)
\equiv
\left(
\begin{array}{c}
q \cdot r\\1
\end{array}
\right)
=
|q \cdot r\rangle
\eeq

Thus it is reasonable to write things like $\delta q = |q_1\rangle -
|q_2\rangle$ and $|q_2\rangle=|q_1\rangle+\delta q$.  We will often
write the latter as $|q_1+\delta q\rangle$ without confusion,
and will often switch between additive ($|q+\delta q\rangle$) and
multiplicative ($\Delta(\delta q)\ |q\rangle$) representations of addition.  

\section{Totalisation and rollback}

Obviously, some concept of rollback is possible in the above system
only if the effect of each operator can be undone. This is only 
possible if the operators form a group so that inverses always exist: 
\begin{lemma} 
The span of $\mu(F_X) \cup \Delta(F_X) \cup C(F_X)$ form a multiplicative 
group iff $F_X$ is 0-totalized. 
\end{lemma} 
\begin{proof} 
An arbitrary operator has the form
\beq
\left(
\begin{array}{cc} 
a & b \\
0 & 1
\end{array} 
\right)
\eeq 
so that it maps $|X\rangle$ to $|a \cdot X + b\rangle$. The inverse operation for this mapping (if it exists) maps $|Y\rangle$ to $|(Y-b)/a\rangle$.
This always exists iff it exists for $a=0$, which in turn holds
iff $F_X$ is 0-totalized. 
\end{proof} 
The operators that are not already invertible in non-totalised fields
include $C(q)$ and $\mu(0)$; inverting these requires an explicit division by zero. 

The work of Bergstra and Tucker views field totalisation as an information 
problem. Totalising a field is a matter of remembering -- for each value 0  -- h
how it came about, so that dividing and multiplying by 0 can be 
accomplished algebraically. 
E.g., the fact that $0 \cdot X=0$ means that for the latter $0$, $0/0 = 0 \cdot X/0 = X \cdot (0 / 0) = X \cdot 1 = X$. This can only be done if we somehow
remember that the $0$ was generated via the operation $0 \cdot X$. 
Division-safe calculation refers to the situation in which this kind of 
history is {\em not} needed. 
One of the key contributions of Bergstra and Tucker is to demonstrate just 
how difficult it is to determine whether information is being lost. 

For us, rollback is similarly an information problem. When we employ the operator $C(q)$, we lose information. Zero-totalizing the base field $F_X$ is equivalent with saving that (algebraic) information for later use, to compute $C^{-1}(q)$. 
In like manner, ``division-safe'' calculation is analogous to ``reversal-safe'' 
rollback, in the sense that the concepts of safety are equivalent. A rollback is
reversal-safe exactly when the operations that led to the current state are division-safe. 

Our problem is subject, however, to more forms of information loss than zero-totalisation can correct. To model other losses, we must consider issues of time
and determinism. 

\section{Modeling changes over time}

So far we have a non-temporal model of change; however, since
operators $C(q)$ do not commute, partial orderings of operations are
important and time is the natural expression of sequence at our system level
of abstraction. Using our framework we can now say that between any two times $t$ and
$t+\delta t$, a system state $|q\rangle$ might change from $|q\rangle$
to $|q+\delta q\rangle$\cite{burgesstheory}.

Starting at a low level, change can be modeled by 
a finite automaton, where the transitions are operators $\hat O$ as above. 
The changes applied after a finite series of steps can be represented
as a matrix product of the form $\hat O_1 \cdots \hat O_n$. 

In automaton theory, one makes the distinction between
deterministic and non-deterministic automata\cite{lewis1}. A
deterministic automaton is a 5-tuple 
\beq
M_D = (Q,A,|q_i\rangle,Q_f,\Delta_D),
\eeq
where $Q$ is a set of states, $A$ is an alphabet of input
instructions, $q_i$ is an initial state, $Q_f \in Q$ a set of possible
final states and $\Delta_D: Q\times A\rightarrow Q$ is a transition
function that takes the automaton from a current state to its next
state, deterministically in response to a single input symbol from the
alphabet. Such a string of operational symbols is the basis for a `journal'
that is intended to track the changes made.
In automata which form `sufficiently dense' graphs, the transition
function's effect may also be seen as an evolution operator, driving the
system through a path of states
\beq
\Delta_D(I): |q\rangle \rightarrow |q+\delta_I q\rangle
\eeq
This mapping might or might not be a bijection; it might or
might not possess an inverse. Although automata are considered to be
a model for computation or even grammars, they can be used to describe
change at any `black box' level of system description, as the model is
entirely general.

A non-deterministic automation is almost the same as a deterministic
one, except that its transition function $\Delta_N: Q\times \{A\cup
0_+\} \rightarrow Q$, accepts one more pseudo-symbol, $0_+$, which is
the empty input string. Thus, a non-deterministic automaton can make
transitions spontaneously, unprompted by input. In the language of \cite{burgesstheory}, a
deterministic automation is a {\em closed system} and a
non-deterministic one is an {\em open system}.

Non-deterministic changes are common in real systems. 
Examples of non-deterministic changes include:
\begin{itemize}
\item Delete of all files from a computer.
\item Remove a firewall, system is infected by virus, replace firewall, system is
still infected by virus.
\item Checkpointing: erase state and replace with a stored image
from time $t_0$, all history is lost between $t_0$ and now.
\end{itemize}

Note that there is a one-to-one
correspondence between input and output only in the deterministic
case. However, there are very few closed deterministic automata in
real-world computing environments. Networks of users operating
multi-tasked, multi-threaded applications create the high level
appearance of many overlapping non-deterministic automata.

\begin{definition}
We define a `computer system' to be a non-deterministic automaton,
represented as a set of states with types, and data sets that are isomorphic 
through extension to the rational numbers. 
\end{definition}
One must assume non-determinism of all actual systems, because in any
modern, preemptive operating system high level changes to observable
data objects cannot be traced to an alphabet of intentionally applied
and documented operations, thus there are apparent transitions that
cannot be explained by a journal.

\section{Journals and histories of change}

In any solution generated by a difference equation (or transition
function), the conversion of small increments or `deltas' into an
absolute state requires the specification of end-points, analogous to
the limits of a contour integral in calculus along a well-defined path
$P$:
\beq
|q_i\rangle - |q_f\rangle = \int_{P\,i}^f dq.
\eeq
This path corresponds to a sequence of input symbols for an automaton, 
corresponding -- in our case -- to operators to be applied. 
The analogue in terms of group transformations is to start from an origin
state, or `ground state' $|0\rangle$ (often called a baseline state in system
operations), and to apply relative changes sequentially from this to
achieve a final desired outcome. 

The choice of the baseline state lies outside of the
specification of the change calculus. The origin or baseline state
is an ad hoc fixed point of the system, by virtue of an
external specification alone. It is arbitrary, but usually plays a
prominent role in system operators' model of system change. 
In this work, the choice of a baseline state is part of what we shall
refer to as a calibration of the system, but counter to tradition we
shall advocate calibration of the {\em end state} rather than the ad
hoc {\em initial state}.

To model intended versus actual change, we introduce the notion of a
journal, inspired by the notion of journaling in filesystems. A
journal is a documentation of changes applied to a system
intentionally, noting and remembering that  -- in real systems -- 
this can be different from what actually takes place.

\begin{definition}[Journal $J$]
A complete, ordered sequence of all input symbols passed to an
automaton $\alpha^*$ from an initial time $t_i$ to a final time $t_f$
is called the automaton's journal $J = (t_i,t_f)$. Each symbol
$\alpha$ corresponds to a change in system state $\delta_\alpha q$. A
journal has a scope that is known to the user or process that writes
the journal. A journal change $\delta J$ involves adding or removing symbols
in $\alpha$ to $J$, and adjusting the times.
\end{definition}

Two journals $J_1$ and $J_2$ may be called {\em congruent} if they
have the same number of symbols $|J_1|=|J_2|$ and every symbol is
identically present and in the same order\cite{lisa0299}.
\begin{lemma}
The final state $|q_f\rangle$ obtained by applying congruent journals of transitions $J_1, J_2$
to identical automata $M_1,M_2$ is identical, iff the initial
states $|q_i\rangle$ are identical, and $M_1$ and $M_2$ are deterministic.
\end{lemma}
This follows from the definitions of (non-)deterministic automata
which allows spontaneous changes $\delta_{0_+}q$. To record all
changes in a non-deterministic system we need to record absolute state even when
no input change is made. This brings us to:

\begin{definition}[History]
A complete, ordered stack of all intermediate snapshots of a
system's total state $|q_f\rangle(t)^*$ output by an automaton at all times
$t$ between an initial time $t_i$ to a final time $t_f$ is called the
automaton's history $H$. A change $\delta H$ involves pushing or popping
the complete current state onto the stack $H$, and adjusting the times.
\end{definition}
The history $H$ is capable of including states that were not directly
affected by the journal transitions $\delta_\alpha q$.
We use a stack as a convenient structure to model histories; see for
instance \cite{bergstra3} and references for a discussion of stacks.
The ability to model system configuration by relative changes is
affected by the following lemma:

\begin{lemma} \label{l1}
For any automaton $M$, $|H_M| \ge |J_M|$, and $|H_M|=|J_M|$ iff $M$ is
a deterministic automaton (closed system).
\end{lemma}
The proof follows from the form of the transition functions
for automata, and the possibility of one or more occurrences
of $0_+$ in the input of a non-deterministic automaton. In a deterministic system
each $\alpha$ leads to a unique labeled transition $\delta_\alpha q$, and vice versa.
In the non-deterministic case, the history can contain any number of changes $\delta_{0_+}$
in addition to the $\alpha$, thus the length of the history is greater than or equal to
the length of the journal.

A journal is thus a sequence of {\em intended} changes, whereas a history is a sequence of actual changes.

\begin{definition}[Roll-back operation $J^{-1}$]
The inverse application of a string of inverse journal operations 
is called a roll-back operation.
The inverse is said to exist iff every operation symbol in the journal has a unique inverse.
\end{definition}
For example, for relative change:
\beq
J(q,q') : |q\rangle &\mapsto& | q+\delta q_1+\delta q_2+\delta q_3\rangle \equiv |q'\rangle
\eeq
and
\beq
J^{-1}(q,q'): |q'\rangle &\mapsto& | q'-\delta q_3-\delta q_2-\delta q_1\rangle \equiv |q\rangle
\eeq

\begin{lemma}
A roll-back journal $J^{-1}(q_i,q_f)$, for automaton $M$, starting
from state $|q_f\rangle$ will result in a final state $|q_i\rangle$
iff $M$ is deterministic and $J^{-1}$ exists.
\end{lemma}
\begin{proof}
Assume that $M$ is non-deterministic; then the transition to state
$q_f$ is only a partial function of the journal $J$, hence $J^{-1}$ 
has more than one candidate value and thus cannot exist. If $M$ is deterministic then
the inverse exists trivially by construction, provided that each operation
in the journal exists.
\end{proof}
Setting aside technical terminology, the reason for a failure to
roll-back is clearly the loss of correspondence between journal and
history caused by changes that happen outside the scope of the intended
specification. This loss of correspondence can happen in a number of
ways, and (crucially) it is likely to happen because today's computer
systems are fundamentally non-deterministic\footnote{
In \cite{burgesstheory}, it was pointed out that this mirrors results
in information theory\cite{cover1} about transmission of data over
noisy channels, for which one has the fundamental theorem of channel
coding due to Shannon\cite{shannon1} that enables the re-assertion of
correspondence between a journal (transmitted data) and actual history
(received data) over some time interval. However, we shall not mix metaphors by pursuing this point
here.}.

\begin{definition}[Commit and Restore operations]
A commit operation at time $t$ is a system change $\hat g$ followed by a push of current
history state onto a stack as consecutive operations:
\beq
commit(t) : (\hat g,push(|q\rangle(t)))
\eeq
A restore operation is a sequence of one or more operations:
\beq
restore(t): pop(|q\rangle)
\eeq
\end{definition}
These operations are typical of version control schemes, for example.
The importance of this construction is that previous states can be
recaptured regardless of whether the operation $\hat g$ is invertible or not.
\begin{lemma}
For automaton $M$, $n$ consecutive restore operations starting
from $t_f$, are the inverse of $n$ consecutive
commit operations ending at $t_f'$, iff the journal of changes between
 $t_f > t_f'$ and $t_f'$ is empty and $M$ is a deterministic automaton.
\end{lemma}
The proof, once again, follows from the absence of uncaptured
changes. If $t_f'>t_f$ and the journal is empty then the only changes
that can have occurred come from symbols $0_+$, but these only occur
for non-deterministic $M$.
We add the following to this:
\begin{lemma}
A system journal $J$ cannot be used to restoring system state for arbitrary
changes $\hat g$.
\end{lemma}
This result is clear from the independence of the restore operation
on $\hat g$, and the lack of a stack of actual state in a journal $J$.

What the foregoing discussion tells us is that there is no predictable
outcome, either in a forward or a reverse direction, in an open
(non-deterministic) system using relative change, and that a journal
is quite useless for undoing changes that have no inverse. System
configuration is analogous to making calculations in which variables
change value spontaneously (as in fact they do without error
correction at the hardware level). To make change computation
predictable, we need to fix the outcomes rather than the sequences of
operations, using `singular change operations' for computing the final
state. This was the main observation learned in the development of
Cfengine\cite{burgessC11,burgessC1,couchDSOM2003}.

\section{Singular transitions and absolute change}

In the foregoing cases, the initial choice of state $|q_i\rangle$ was external to the
specification of the change, and was the `origin' of a sequence
of changes in a journey from start to finish. This relative (`sequential
process') approach to change is deeply in-grained in management and
computing culture, but it fails to bring the require predictability due to
underlying system indeterminism. The problem is the reliance on the $+$
operation to navigate the state space, so our next step is to suppress it.

Now consider a class of transitions that are not usually considered in classic
finite state machines.  These are (non-invertible) elements $\hat p$ with the property
%%% $\Delta_{N,D} \rightarrow g$, where 
that $\hat p |q\rangle = |q_0\rangle$, for any
$q$. The final states are `eigenstates' of these singular group operations:
$\hat p|q_0\rangle = |q_0\rangle$. These effectively demote the explicit reliance on
$+$ and replace it with a linear function.

\begin{definition}
A singular transition function $C_{|q_0\rangle}$ is a transition from any state $|q\rangle$
to a unique absorbing state $|q_0\rangle$. It is a many-to-one transition, and is hence
non-invertible without a history.
\end{definition}

Such transition functions (operators) were introduced in \cite{burgessC1} and described in
\cite{burgessC11}, as an alternative to relative change to restore
the predictability of outcome. These
`convergent operations' are based on fixed points or eigenstates of a
graph. They harness the property of zero elements to ignore the
current and historical states and to install a unique state regardless
of the history or the determinism of the system. Such parameterized
operators form a semi-group $C_{|q_0\rangle}$ with the abstract
property:
\beq
C_{|q_0\rangle} |q\rangle &=& |q_0\rangle\nonumber\\
C_{|q_0\rangle} |q_0\rangle &=& |q_0\rangle.\label{convergence}
\eeq
For ease of notation in the following, we drop the $|q_0\rangle$ subscript and 
write $C$ 
for $C_{|q_0\rangle}$

The price one pays for this restoration of predictability is an inability to
reverse the change. Let us suppose that an object
$C^{-1}$ exists such that $C^{-1} C = I$, satisfying the latter equation. Then operating
on the left, we may
write using (\ref{convergence}):
\beq
C \;|q_0\rangle &=& |q_0\rangle\nonumber\\
C^2 \;|q_0\rangle &=& C \;|q_0\rangle\nonumber\\
C^{-1} C \;|q_0\rangle &=& C^{-1} \;|q_0\rangle
\eeq
Thus, at $|q_0\rangle$ we have idempotence and a constraint:
\beq
C \;|q_0\rangle&=&  C^{-1}\;|q_0\rangle\\
C \;|q_0\rangle&=&  C^2\;|q_0\rangle\label{xidemp}.
\eeq
The latter result (\ref{xidemp}) is independent of the existence of an inverse.
For a ring, this condition is equivalent to the `restricted inverse law' used in
\cite{bergstra2,bergstra4}, and it tells us that the inverse would have to be
either $0$, $1$ or $+\infty$.
\begin{lemma}
The operators $C_{|q_0\rangle}$ are idempotent and converge on a fixed point final state $|q_0\rangle$.
\end{lemma}
This follows immediately from eqn (\ref{xidemp}).  The value of these
operations is that they can be iterated endlessly, with predictable
outcome, in the manner of a highly compressed system error-correction
process.

Example 1: One can view the state $|q\rangle$ as embodied in the 
operator $C_{|q\rangle}$ and thus view $C_0$, 
$|q\rangle$, and $|q_0\rangle$ as elements of the same semigroup. 
Then we may write: 
\beq
C_0 |q\rangle &=& |q_0\rangle\\
C_0 |q_0\rangle &=& |q_0\rangle.\label{ringex}
\eeq
Assuming an additive inverse for each element (in the statespace), 
and subtracting these
equations for arbitrary $q$ leads to the conclusion that $C_0 = |q_0\rangle =
|0\rangle$, thus there is only a single object with this ability to take an
arbitrary initial state and render a predictable outcome. Note that, in this
representation, $C$ and $|q\rangle$  belong to the same semigroup of scalars. So, choosing
$|x\rangle = C_0 = |q_0\rangle$, we service
\beq
(x^{-1} \cdot x) x = x
\eeq
using (\ref{ringex}). This is the restricted inverse law for fields\cite{bergstra2,bergstra4}.

The zero plays a fundamental role as an eraser.
The uniqueness of zero is not an impediment to using the zero element
as a `policy operator' which sets an intended state, as we are free to
construct a homomorphism $h(q)$ which {\em calibrates} or shifts the
absolute location of the solution: e.g.
\beq
C_0 h(q) &=& h(q_0)\nonumber\\
C_0 h(q_0) &=& h(q_0)\label{homo}
\eeq
to shift the calibration point from $q_0$ to $q_0^*$. 
%%% where $q_0^*$ represents the new calibration point and $h(q) = q-q_0^*$, this
%%% yields the solution $C_0 = 0, q_0 = q_0^*$.

%%% I don't believe this: 
%%%We can make a similar version of this as a vector space when $C$ and $q$
%%%belong to different rings.

Example 2: Consider the tuple form used earlier, and let

\beq
C_0 &\mapsto& \left( 
\begin{array}{cc}
0 & q_0\\
0 & 1\\
\end{array}
\right)\nonumber\\
|q_0\rangle  &\mapsto& \left( 
\begin{array}{c}
q_0\\
1
\end{array}
\right)\label{matrixrep}
\eeq
Thinking of these as elements of the same semigroup and 
subtracting these equations leads to a result that is identically
true, hence we are free to choose the value of $q_0$ as a matter of
policy.  However, one observes that $C_0$ does not 
possess a defined inverse according to the normal rules of fields.

%%% \begin{lemma}
%%% Any inverse representation of a linear operator, taking values
%%% in a field $F$ and satisfying 
%%% $C^{-1}C = I$ and (\ref{convergence}) must
%%% involve division by zero.
%%% \end{lemma}
%%% \begin{proof}
%%% Let $C$ be a linear function of an underlying field, then $C$ has 
%%% the form $C = m q + q_0$,
%%% and thus the inverse is $C^{-1} = (C-q_0)/m$.
%%% In order to satisfy (\ref{convergence}) for all $q$, we must have $m = 0$.
%%% \end{proof}
%%% This shows that the `rollback' of an absolute state change is directly
%%% analogous to a division by zero in a ring computation. The reason is
%%% clear: both are attempts to reverse a transition after dumping all
%%% history of initial state. We must therefore conclude that either such
%%% operations are irreversible, or that an inverse must be constructed
%%% for completeness, subject to other limitations.

\section{Computing and reversing absolute states}

Fields have only one element with singular properties: the zero
element. It plays two distinct roles: as an identity element for
the $+$ operation, and as a fixed point in scaling under $\cdot$. As a
fixed point, zero annihilates state, since $0 q = 0$ for any $q$. The
zero element thus ignores and deletes any history that led us to the
state $q$.

It is useful to think of the $C$ operators in the above as a kind of 
zero-element: they annihilate state in a similar way. 
The utility of the convergent operations for bringing about absolute
change is such that it is useful to embed them in the formalism of a
general field structure for computation.
One motivation for this is the recent work by Bergstra and Tucker
of totalization of fields, and `Meadows', in which they replace the
partial function (excluding $0$ for division) at the heart of field computation 
with one that is total, up to constraints. We find their construction
intriguing and highly relevant to the matter of reversibility of state. As in Bergstra and Tucker, we reason by first defining the algebraic signatures of structures. 

We construct an image of a field $F$, with initial algebra $Alg(\Sigma_F,E_F)$ (as yet a
regular field), by introducing a map in three piecewise partial representations $\Phi(F) =
\{C(F),\Delta(F),\mu(F)\}$.
The signature contains only product explicitly, 
as addition is concealed as described in section \ref{pt}:
\beq
\cdot &:& \Delta \times \Delta \rightarrow \Delta\\
      &:& \mu \times \mu \rightarrow \mu\\
      &:& C \times C \rightarrow C\\
I_\Delta &=& \Delta(0) \in \Delta \\
I_\mu &=& \mu(1) \in \mu \\
\mu(0) &\not\in& \mu\\
\eeq
We define $E_\Phi$ as the image of $E_F$ for the field by,
\beq
\Delta(x)\Delta(-x) &=& I_\Delta, ~~~~~\forall x \in F\\
\mu(x^{-1})\mu(x) &=& I_\mu ~~~~~\forall x\not = 0 \in F\\
\Delta(x)\Delta(y) &=& \Delta(y)\Delta(x)~~~~~\forall x,y \in F\\
\Delta(x)(\Delta(y)\Delta(z)) &=& (\Delta(x)\Delta(y))\Delta(z)~~~~~\forall x,y,z \in F\\
\mu(x)(\mu(y)\mu(z)) &=& (\mu(x)\mu(y))\mu(z)~~~~~\forall x,y,z \not 0 \in F\\
C(x)C(y) &=& C(x)  ~~~~~\forall x \in F\\
C(x)(C(y)C(z)) &=& (C(x)C(y))C(z)~~~~~\forall x,y,z \in F\\
\eeq
An example in the matrix representation is given by:
\beq
C(x) &=& \left(
\begin{array}{cc}
0 & x\\
0 & 1
\end{array}
\right)\\
\Delta(x) &=&
 \left(
\begin{array}{cc}
1 & x\\
0 & 1
\end{array}
\right)\\
\mu(x) &=& \left(
\begin{array}{cc}
x & 0\\
0 & 1
\end{array}
\right)
\eeq
The function $C$, in any representation, 
gives us a way of representing absolute, not relative, changes of
state.  This is an important ability in maintaining order in a system,
and it is the basis on which Cfengine\cite{burgessC1} operates on
millions of computers around the world today. Each operation is a
function of a field, in which the zero element is mapped to a desired
state. The set of all possible parameterized $C(F)$ must therefore span a field,
and yet it contains no (multiplicative) inverses at all. This is the interesting
paradox which plays into the work of Bergstra and Tucker.
The restriction $x \not = 0\in F$ is prominent.

It is not our intention to reiterate the arguments for totalizing
fields, presented by Bergstra and Tucker\cite{bergstra4}.  As they
point out, there is a number of ways to restore `faith' in the
connection between state and history of change after a zero operation,
using proof systems, axioms and algebraic properties. Each brings a
different kind of merit. As they remark, the issue is not so much
about going backwards (reversal) as about going forwards in a way that
is unaffected by an ill-defined attempt at reversal.

One remedy relies on proving the outcome of a change was not affected
by the result of $0^{-1}$, i.e. the final state is independent of the
path taken to evaluate it. Another involves changing the definitions of
computation (change) to disallow unsafe operations. Finally one might
simply give up on certain requirements so that the outcome satisfies a
well defined set of equations (policies) in order to prove that the
result is well-defined.
Here, we observe by analogy that one may:
\begin{itemize}
\item Introduce a stack of history snapshots
to some maximum depth\cite{bergstra3}. (This is difficult to do in arithmetic but it is plausible for some system changes.)
\item Totalize the data type, using the notion of a totalized field, e.g. set $0^{-1}= 0$, or equivalently, $C^{-1} = C$.
\item Perform a naive reversal and then apply some policy equational specification
to clean up the result.
\item Abandon the attempt to introduce reversals altogether (``rollback does not exist'').
\end{itemize}

\section{Calibration of absolute state}

Let us complete the abstract formalization of the operators for
absolute change, which builds linear functions on top of the totalized
field approach of Bergstra and Tucker, and ends with a vector
space. We no longer care about $\Delta$ and $\mu$, but want to embrace
the properties of the zero operators to bring predictability in a
non-deterministic environment.  We start with a signature for the
convergent operators based on a different use of commutative rings as a parameterization of the
zeroed outcome, and end with
non-commutative, non-invertible representations. Let $F$ signify a field,
with the usual field axioms.

We use a $\Sigma$-algebra $\Sigma_C = \{ |q\rangle\, |\,
I,\,0,\,\oplus,\,\circ,\, C \}$, and this is understood to extend
the field algebra $Alg(\sigma_F,E_F)$. Thus, for any index set labels
$\alpha,\beta$, labeling the underlying field $Q$, we have signature:

\beq
{\rm Symbols}: &~& I, 0, | q_\alpha\rangle, q\in Q\\
{\rm Operations}: &~&\\
&~& I \rightarrow F\\
&~& 0 \rightarrow F\\
 C: &~&F \rightarrow F'\\
\oplus: &~& C(F\times F) \rightarrow  C(F)\\
\circ: &~& C(F)\times  C(F) \rightarrow  C(F)\\
&~& C(F)\times F' \rightarrow F'\\
{\rm Equations} (E_C):&~&\\
&~& C_\alpha |q\rangle = C(q_\alpha) |q\rangle = |q_\alpha\rangle\label{zero}\\
&~& C_\alpha \circ  C_\beta =  C_\alpha\label{idemp}\\
&~&( C_\alpha\circ  C_\beta)\circ C_\gamma =  C_\alpha\circ(  C_\beta \circ C_\gamma)\\
&~& C_\alpha\oplus  C_\beta =  C_{\alpha+\beta}\\
&~&( C_\alpha\oplus  C_\beta) \oplus  C_\gamma =  C_\alpha\oplus ( C_\beta \oplus  C_\gamma) 
\eeq
Naturally, these are true for all $\alpha,\beta,\gamma$, and
we are working with $Alg(\Sigma_F\cup\Sigma_C,E_F\cup E_C)$.  The
opacity of formalism belies a simple structure. Every state
$|q\rangle$ is fully specified by a field value $q \in
F$. Similarly, every convergent operator $C(q)$ is fully
specified by a field value $q \in  F$, and results in a new
value $q \in F$, which obeys the zero property
(\ref{zero}). Thus the $C$ is a transformer which takes any input
state and outputs a specific state given by its label (but
importantly, only one at a time). This has the `zero' property of
ejecting initial state and replacing it wholesale with particular one.
Clearly, the operators must be idempotent from (\ref{idemp}).

The representation in (\ref{matrixrep}) is useful to see how a
tuple-representation quickly captures this algebra.  We say that the
repeated operation of an operator $C(q_0)$ `converges', as it always
returns the system state to its fixed point $|q_0\rangle$. This algebra describes
the behaviour of a single `convergent operator' or `promise' in
Cfengine\cite{burgessC11,burgessC1}.  We cannot define $C^{-1}$
because the symbol $0^{-1}$ is not defined in the underlying field
$F$, but we may totalize the field\cite{bergstra4} with corresponding
merits and conditions to assign a meaning to a reversal or `roll-back'.

\section{Re-calibration - change of policy}

There is only a single fixed point for each operator $C(q_0)$.
What happens when we want to change the outcome of a `promised state',
i.e. change the value of $q_0$?
The homomorphism $h$ on states, in eqn. (\ref{homo}) allowed us to
calibrate a single singular outcome to any field value by shifting the
zero, but this is less useful than modifying the operators themselves
to bring about the desired result. This transformation then has the simple
interpretation as an operator the re-calibrates the system baseline.

\begin{lemma}
Each operator has only one singularity, i.e.
let $F$ be a field, totalized or not, and let $0_1$ and $0_2$
be zero elements for $\cdot$, then $0_1 = 0_2$.
\end{lemma}
The proof follows by substitution of the field axioms: $0_1
x = 0_1$, $0_2 x = 0_2$, setting $x = 0_2$ in the former, implies
$0_1=0_2$. Hence the zero element is unique in a field. 

This means that
we cannot have more than one policy fixed point per field. In
configuration terms, one cannot have more than one policy for a data
item, so any path of changes parameterized by chaining $q_0$ can be
uniquely characterised.

This leaves only the possibility of shifting the fixed point by
re-calibration, or change of policy. This is no longer a journal
of deltas, but a kind of `teleportation' or `large transformation'
in the group theoretic sense.
Given this, and the utility of
formulating policy changes as applied operations, it is useful to
reformulate the values in terms of {\em vector spaces}. The
specification of a vector space is somewhat similar to that of a ring
or field except that it is not automorphic. 

Let $F$ be a field (totalized or not). A vector space of $F$ is a
triple $(S,+,\cdot)$, with the equations:
\beq
+ &:& S\times S \rightarrow S,\nonumber\\
\cdot&:& F\times S \rightarrow S,\nonumber\\
0 \in S &|& x+0 = 0+x = x,  ~~~~ \forall x \in S\nonumber\\
-x \in S &|& x+(-x) = (-x)+x = 0,  ~~~~ \forall x \in S\nonumber\\
x+y &=& y+x,  ~~~~ \forall x,y\in S\nonumber\\
(x+y)+z &=& x+(y+z),  ~~~~ \forall x,y,z \in S\nonumber\\
(\alpha\beta)z &=& \alpha(\beta z), ~~~~ \forall \alpha,\beta \in F, z \in S\nonumber\\
1_F \in F &|& 1x = x1 = x,  ~~~~ \forall x \in S\nonumber\\
(\alpha+\beta)x &=& \alpha x + \beta x,  ~~~~ \forall \alpha,\beta \in F, x \in S,\nonumber\\
\alpha(x+y) &=& \alpha x + \alpha y,  ~~~~ \forall \alpha \in F, x,y \in S,
\eeq

The usefulness of this map is that it involves an external `promise' or `policy'
field $F$ from which we may construct the set of $C_\alpha$, not merely
an automorphic image of a single set. Thus we can separate policy from changes with convergent, fixed-point
zero-operators $0_A \in F_A$, all acting on a single set of states $q
\in S$.  We thus arrive, by a different route, at the formulation as a
vector space in ref. \cite{burgesstheory}.

We note finally that a change of calibration cannot be a commutative ring.
\begin{lemma}
Let $C_A$ and $C_B$ be zeros of $F_A$ and $F_B$. Then $C_A$ and $C_B$
cannot commute unless $A=B$.
\end{lemma}
\begin{proof}
The proof is similar to the earlier proof of uniqueness of zero in a ring. 
We have $ C_A q =  C_A$, and  $ C_B q =  C_B$ for all $q$. Substituting $q = C_B$
in the former, we have
\beq
C_A(C_B)q &=& C_A C_B = C_A\nonumber\\
C_B(C_A)q &=& C_B C_A = C_B
\eeq
Thus the commutator
\beq
[C_A, C_B] =  C_A  C_B -  C_B  C_A =  C_A -  C_B \not= 0 ~~~(A \not= B)
\eeq
This proof does not depend on the representation of $F_A$ and $F_B$,
thus it applies equally to higher dimensional tuple formulations also.
\end{proof}

\section{Predicting outcome with roll-back-safe change}

The problems of indeterminism cannot be addressed without absolute
change operations, but these do nothing to repair the problem of
unsafe reversals. The $C$ operations allow us to basically forget
about indeterminism, but not irreversibility.  We therefore need to
find an approach analogous to that of \cite{bergstra4} during
non-commutative strings of system re-calibrations. We have one
advantage here: a lack of commutativity. This is in fact a strength as
it makes the need for reversal practically irrelevant.  In our view,
there is then only one natural choice for $C^{-1}$ or $J^{-1}$ and
that is to apply or re-apply the current policy $C(t)$: it is
absolute, idempotent and it overrides any previous `mistakes'.

We have also one disadvantage compared to \cite{bergstra4} and that is
that time is relevant: we cannot undo the potential consequences of
being in a bad state unless we manage to totality of state within the
system. For real computers, that might be almost the entire Internet
(e.g. during the spread of viruses).

Other weaker arguments can be made for resetting state to a baseline, e.g.
\noindent (i) Use an arbitrarily chosen baseline state $|q_{\rm initial}\rangle$
or $|t_0\rangle$ so that an arbitrary journal of convergent changes $\hat J_0$
\beq
|t_{\rm final}\rangle &=& \hat J_0 | t_{\rm initial}\rangle\nonumber\\
&\equiv&\ldots \hat O_2\hat O_1  | t_{\rm initial}\rangle
\eeq
has an inverse such that
\beq
J_0^{-1} |t_{\rm final}\rangle =  | t_{\rm initial}\rangle.
\eeq
Assuming the existence of an operator $\hat O_{\rm initial}$ such that
$\hat O_{\rm initial} |q\rangle = |t_{\rm initial}$, then clearly
\beq
J_0^{-1} = \hat O_{\rm initial}.
\eeq
These two choices are both forward-moving absolute changes since they
both involve an arbitrary decision and they both move forward in time.
However the latter is less natural, since it affects to return to a
time in the past which might have nothing directly to do with where
one needs to be in the present.  Our study was motivated by
predictability.  The principal advantage of these remedies lies in
knowledge of the outcome, in the absence of a complete specification.

\section{Multi-dimensional operators}

In the discussion above, we have restricted ourselves to the maintenance
of a single scalar system-value. The issue of dependencies amongst
system changes enters quickly as the complexity of layered models of a system
grows. It was shown in \cite{burgessC11} that one can develop a spanning
set of orthogonal operations that covers the vector space like a coordinate
system, simply by embedding in a geometrical tuple-fashion. In the simplest
expression, one sees this by extending the matrix representation to higher dimensions.

One way to do this is to consider a system as a controlled by a vector of
its individual configuration parameters $X_i$, where each parameter is embedded into the field of rationals and encoded in the obvious way: 
\beq
|X\rangle = 
\left(
\begin{array}{c}
X_1 \\
X_2 \\
\vdots \\
X_n \\
1
\end{array} 
\right)
\eeq
where $|X\rangle$ is `Dirac notation' for state. 
As optimistic and large as current systems may be, they remain finite and can be modeled by finite vectors. 
We define relative operators for individual parameters in a state as 
\beq
\Delta_i(q) = 
\left(
\begin{array}{cccccc}
1 & 0 & \cdots & & & 0 \\
0 & 1 & & & &  \\
\vdots &  & \ddots & & & \\
& & & 1 & & q \\
& & & & \ddots & \\ 
0 & & & & & 1 
\end{array} 
\right) 
\eeq
(where $q$ appears in the $(n+1)$st column of the $i$th row). 
Likewise, absolute operators are defined as 
\beq
C_i(q) = 
\left(
\begin{array}{cccccc}
1 & 0 & \cdots & & & 0 \\
0 & 1 & & & &  \\
\vdots &  & \ddots & & & \\
& & & 0 & & q \\
& & & & \ddots & \\ 
0 & & & & & 1 
\end{array} 
\right) 
\eeq
(where $q$ appears again in the $(n+1)$st column of the $i$th row), 
and multiplicative operators as 
\beq
\mu_i(q) = 
\left(
\begin{array}{cccccc}
1 & 0 & \cdots & & & 0 \\
0 & 1 & & & &  \\
\vdots &  & \ddots & & & \\
& & & q & &  \\
& & & & \ddots & \\ 
0 & & & & & 1 
\end{array} 
\right) 
\eeq
(Where $q$ appears in the $i$th row and column). 
Thus we propose to model configuration changes in a system via a set of 
matrices with rational entries. 

This completes the construction of the Cfengine operators. Clearly
the zero inverse solution applies independently to each of these
diagonal operators in this basis, but becomes rapidly more entangled in other
parameterizations, where dependencies occur.

\section{Concluding remarks}

We have shown that neither the outcome of a journal of changes,
nor a reversal (undo operation) is generally meaningful or
well-defined in an incompletely specified, or non-deterministic
system. A deterministic outcome can only be obtained by grounding
a system to a policy defined state, analogous to `zero' in a field.

Neither the `restoration of state by roll-back' nor
`division safe calculation', \`a la Bergstra and Tucker, are about
how one goes backwards, but rather about how one recovers meaningfully
forwards, given that a poorly defined operation was attempted at some
point in the past.  A classic answer is `well, don't do that' -- but
we know that someone will always attempt to perform ill defined
operations and thus our story has a practical meaning, of some
importance.

The solutions here mirror way the problem of singularities is handled in other
areas of mathematics, e.g.  in complex analysis one as analytical
continuation\cite{continuation}, in which a path or history through
the states can be defined such that the final result avoids touching
the singular cases.  Similarly in algebraic topology, the uniqueness
of the result can then depend on the path and cohomology.

The totalization remedies described by Bergstra and Tucker underline
an approach to a wider range of problems of incomplete
information. The ultimate conclusion of this work is that `rollback'
cannot be achieved in any well-defined sense without full system
closure. A choice about how to go forward is the only deterministic
remedy.  

There are plenty of topics we have not touched upon here.
\begin{problem}
We have not taken into account fields in which external boundary values
are imposed on $S$. Then we would have further fixed points in the 
total history of a system to contend with:
\beq
\lbrace \min_S, 0_A, 0_B, \ldots \max_S\rbrace
\eeq
Each of these might be a reasonable candidate for `re-grounding' the system
in an undo. What conditions might be imposed when $0_A$ falls outside the range
$[\min_S,\max_S]$.
\end{problem}

\begin{problem}
We have not taken into account operators that depend on one another
in non-orthogonal fashion\cite{lisa0163}. Dependencies between
operators add potentially severe complications to this account.
\end{problem}

There is a deeper issue with roll-back in partial systems. If a system
is in contact with another system, e.g. receiving data, or if we have
partitioned a system into loosely coupled pieces only one of which is
being changed, then the other system becomes a part of the total
system and we must write a hypothetical journal for the entire system
in order to achieve a consistent rollback.
\begin{problem}
The partial restoration can leave a system in an inconsistent state
that it has never been in before and is not a state that was ever
intended.
\end{problem}

The results in this paper are directly applicable to to hands-free
automation, or `computer immunology', as demonstrated by Cfengine.
Opponents of automation have look for ways of arguing that traditional
journaling approaches to system maintenance are
necessary\cite{lisa0299}, preserving the role of humans in system
repair.  However, we argue that the role of humans is rather in
deciding system policy: it is known that the computational complexity
of searching for convergent operations is in PSPACE and NP
complete\cite{larschap,sunchap}, thus it remains the domain of
heuristic methods and system experts to find these convergent in more
complex cases.

\vspace{1cm}

{\bf Acknowledgment: }
This work is dedicated to Jan Bergstra on the occasion of his 60th
birthday.

\bibliographystyle{unsrt} \bibliography{LISAbib,spacetime}
\end{document}